\documentclass[a4paper,11pt]{article}
\usepackage{color}
\usepackage{fullpage}
\usepackage{amsmath,amssymb,enumerate}
\usepackage{amsthm}
\usepackage{graphicx}
\usepackage{sidecap}
\sidecaptionvpos{figure}{t}
\usepackage{url}
\usepackage{hyperref}
\usepackage{comment}
\hypersetup{colorlinks=true}
\hypersetup{linkcolor=black,anchorcolor=black,citecolor=black,urlcolor=black}
\hypersetup{pdfauthor= {Erik Demaine}{John Iacono}{Grigorios Koumoutsos}{Stefan Langerman}}
\hypersetup{pdfkeywords={Dynamic Optimality}{B-trees}{External Memory Model}{Data Structures}{Competitive Analysis}{Online Algorithms}} 
\usepackage{cite}

\newtheorem{theorem}{Theorem}[section]

\newtheorem{lemma}[theorem]{Lemma}

\newcounter{note}[section]

\newtheorem{corol}[theorem]{Corollary}

\DeclareMathOperator{\OPT}{OPT}

\DeclareMathOperator{\Btree}{B-Tree}
\DeclareMathOperator{\OPTbst}{OPT_{BST}}
\DeclareMathOperator{\OPTbtree}{OPT_{B-Tree}}
\DeclareMathOperator{\IB}{IB}
\DeclareMathOperator{\LB}{LB}

\DeclareMathOperator{\ws}{WS}
\newcommand{\ios}{I/Os}

\makeatletter
\newcommand*{\rom}[1]{\expandafter\@slowromancap\romannumeral #1@}

\newcommand{\lrp}[1]{\left( #1 \right)}

\title{Belga B-trees%
\footnote{This work was supported by the Fonds de la Recherche Scientifique-FNRS under Grant no MISU F 6001 1 and by NSF Grant CCF-1533564.}}

\author{
Erik D.~Demaine\thanks{CSAIL, Massachusetts Institute of Technology.
\texttt{edemaine@mit.edu}}
\qquad
John Iacono\thanks{Universit\'{e} Libre de Bruxelles and New York University.
\texttt{johniacono@gmail.com}}
\qquad
Grigorios Koumoutsos\thanks{Universit\'{e} Libre de Bruxelles. \texttt{greg.koumoutsos@gmail.com}}
\qquad
Stefan Langerman\thanks{Directeur de Recherches du F.R.S-FNRS. \texttt{stefan.langerman@ulb.ac.be}}
}

\date{}

\begin{document}

\maketitle

\begin{abstract}
We revisit self-adjusting external memory tree data structures, which combine the optimal (and practical) worst-case I/O performances of B-trees, while adapting to the online distribution of queries. 
Our approach is analogous to undergoing efforts in the BST model, where \emph{Tango Trees} (Demaine \emph{et al.} 2007) were shown to be $O(\log\log N)$-competitive with the 
runtime of the best offline binary search tree on every sequence of searches.
Here we formalize the B-Tree model as a natural generalization of the BST model. 
We prove lower bounds for the B-Tree model, and introduce a B-Tree model data structure, the Belga B-tree, that executes any sequence of searches within a $O(\log \log N)$ factor of the best offline B-tree model algorithm, provided $B=\log^{O(1)}N$. We also show how to transform any static BST into a static B-tree which is faster by a $\Theta(\log B)$ factor; the transformation is randomized and we show that randomization is necessary to obtain any significant speedup.
\end{abstract}

\includecomment{onlymain}
\excludecomment{onlyapp}

\section{Introduction}
\label{s:intro}
Worst-case analysis does not capture the fact that some sequences of operations on data structures, often typical ones, can be executed significantly faster than worst case ones. Methods of analyzing algorithms whose performance depends on more fine-grained characteristics of the input sequence other than the size $N$ have been coined \emph{distribution sensitive data structures} \cite{thesis,DBLP:conf/birthday/BoseHM13}. Two general methods to bound the performance of such a data structure exist. The first is to explicitly bound the performance by some bound. For binary search trees (BSTs) there is a rich set of such bounds (see e.g.~\cite{DBLP:journals/acta/ElmasryFI13,ChalermsookG0MS16}) like the sequential access bound
\cite{DBLP:journals/combinatorica/Tarjan85}, the working set bound
\cite{DBLP:journals/jacm/SleatorT85,DBLP:conf/soda/Iacono01a}, the (weighted) dynamic finger bound~\cite{DBLP:journals/siamcomp/ColeMSS00,DBLP:journals/siamcomp/Cole00,DBLP:conf/soda/IaconoL16}, the unified bound
\cite{DBLP:journals/tcs/BadoiuCDI07,DBLP:conf/soda/Iacono01a} and many others~\cite{DBLP:journals/algorithmica/BoseDIL16,DBLP:journals/corr/abs-1302-6914,DBLP:journals/corr/abs-1809-01759}.
The other method is to compare the performance of the data structure on a sequence of operations to the performance of the best offline data structure in some model on the same sequence. Such an analysis uses the language of competitive analysis introduced in~\cite{DBLP:journals/cacm/SleatorT85}, where the competitive ratio of an algorithm is the supremum ratio of the performance of the given algorithm to the offline optimal over all sequences of operations over a given length. A data structure which is $O(1)$-competitive in a particular model is said to be \emph{dynamically optimal}~\cite{DBLP:journals/jacm/SleatorT85}. In the BST model, the best known competitive ratio is $O(\log \log N)$, first achieved by Tango trees~\cite{DBLP:journals/siamcomp/DemaineHIP07}. The existence of a dynamically optimal BST is one of the most intriguing and long-standing open problems in online algorithms and data structures (see~\cite{DBLP:conf/birthday/Iacono13} for a survey). The two prominent candidates to achieve dynamic optimality for BSTs are the \textit{splay tree} of Sleator and Tarjan~\cite{DBLP:journals/jacm/SleatorT85} and the \textit{greedy} algorithm~\cite{DBLP:conf/soda/DemaineHIKP09,jltr}, but they are only known to be $O(\log N)$-competitive.

%

\paragraph{{\normalfont{\textbf{Disk-Access Model (DAM).}}}} The \emph{external memory model}, or \emph{disk-access model (DAM)} \cite{DBLP:journals/cacm/AggarwalV88} is the leading way to theoretically model the performance of algorithms that can not fit all of their data in RAM, and thus must store it on a slower storage system historically known as \emph{disk}. This model is parameterized by values $M$ and $B$; the disk is partitioned into blocks of size $B$, of which $M/B$ can be stored in memory at any given moment. The cost in the DAM is the number of transfers between memory and disk, called Input-Output operations (\ios). The classic data structure for a comparison based dictionary in the DAM model, as well as in practice, is the B-Tree~\cite{DBLP:journals/acta/BayerM72}. The B-Tree is a generalization of the BST, where each node stores up to $B-1$ data items, for $B \geq 2$, and the number of children is one more than the number of data items. The B-Tree supports searches in time $O(\log_B N)$ in the DAM, a $\log B$ factor faster than traditional BSTs such as red-black trees~\cite{DBLP:conf/focs/GuibasS78} or AVL trees \cite{MR0156719}.

\paragraph{{\normalfont{\textbf{Dynamic Dictionaries in the DAM.}}}} Here, our goal is to explore dynamic dictionaries in the DAM and to obtain results similar to those known for BSTs.

Surprisingly, prior work in this direction is quite limited. One previous attempt was in the work of Sherk~\cite{DBLP:journals/jal/Sherk95} where a generalization of splay trees to what we call the B-tree model was proposed, but without any strong results. Over ten years later, Bose et. al.~\cite{DBLP:conf/soda/BoseDL08} studied a self-adjusting version of skip-lists and B-Trees, where nodes can be split and merged to adapt to the query distribution by moving elements closer or farther from the root of the tree (here we call this model \textit{classic self-adjusting} B-trees, see Section~\ref{s:model}). They showed that dynamic optimality in this model is closely related to the working set bound. This bound captures temporal locality: for an access sequence $X = x_1,\dotsc,x_m $, it is defined as $\ws(X) = \sum_{i=1}^{m} \log w_X(i)$, where $w_X(i)$ is the number of distinct elements accessed since the last access to the element $x_i$. In~\cite{DBLP:conf/soda/BoseDL08} the authors presented a data structure whose cost is upper bounded by $O(\ws(X)/\log B)$ and obtained a matching lower bound of $\Omega(\ws(X)/\log B)$ for this model, which implies that their structure is dynamically optimal.

Note that the lower bound of~\cite{DBLP:conf/soda/BoseDL08} shows a major limitation of B-trees with only split and merge operations: It implies there are sequences on which they are slower than BSTs. For example, repeatedly sequentially accessing all data items $1,2,\dotsc,N$ requires $O(1)$ amortized time per search for BSTs like splay trees (this is the \textit{sequential access bound}~\cite{DBLP:journals/combinatorica/Tarjan85}) while the lower bound $\Omega(\ws(X)/\log B)$ implies an amortized cost $\Omega(\log_B N)$ in the classic self-adjusting model. In this work, we show that by adding just one more operation, an analogue of the rotation for B-Trees, we can overcome this limitation and obtain significant speedups with respect to standard B-trees.


\paragraph{{\normalfont{\textbf{Our Contribution.}}}} In this work we initiate a systematic study of dynamic B-trees. First, we formally define the (dynamic) B-Tree model of computation (\S\ref{s:model}). Second, we show how to produce lower bounds in the B-Tree model (\S\ref{s:lb}). Then, we introduce a data structure, which we call the \emph{Belga B-Tree}\footnote{The Tango tree was invented 
on an overnight 
flight from JFK airport en route to Buenos Aires, Argentina. The work on the Belga B-Tree has been substantially completed at Cafe Belga, Ixelles, Belgium.}, which is $O(\log \log N)$ competitive with any dictionary in the B-Tree model of computation, when $B=O(\log^{O(1)}N)$ (\S\ref{s:ods}).

 
More generally, we conjecture the following in \S\ref{s:open}: any BST-model algorithm can be transformed into a (randomized) B-Tree model algorithm with a $\Theta(\log B)$ factor cost savings. This would imply that BST model algorithms such as the splay tree~\cite{DBLP:journals/jacm/SleatorT85} or greedy \cite{DBLP:conf/soda/DemaineHIKP09,jltr} would have B-Tree model counterparts, and that a dynamically optimal BST-model algorithm would imply a dynamically optimal algorithm in the B-Tree model.  We leave this conjecture open, but in \S\ref{s:tranform} we do resolve the case of a static (no rotations allowed) BSTs by showing a randomized transformation from a static BST to a static B-Tree such that any algorithm in the static BST model would have factor $\Theta(\log B)$ speedup in the B-Tree model. We also show that no $\omega(1)$-factor speedup is possible for a deterministic transformation in general.






 
 

\section{The B-Tree model of computation} \label{s:model}




 In this section, we define the tree models discussed in this paper. In all cases, we consider data structures supporting searches over a universe of $N$ elements $\mathcal{U} = \lbrace 1,2,\dotsc,N \rbrace$ which we refer to as \textit{keys}. The input is a valid tree $T_{0}$ and request sequence of searches $X = x_1,x_2 \dotsc, x_m$, where $x_i \in \mathcal{U} $ is the $i$th item to be searched.

 \subsection{The BST Model} 
 
 In a Binary Search Tree (BST) data structure, each node stores a single key and three pointers, indicating its parent and its (left and right) children. The key value of a node is larger than all keys in its left subtree and smaller than all keys in its right subtree.  
 To execute each request to search for element $x_i$, a BST algorithm initializes a single pointer at the root (at unit cost)  then may perform any sequence of the following unit-cost operations:
 \begin{itemize}
  \item Move the pointer to the parent or to the left or right child of the current node the pointer points to (if such a destination node exists).
  \item Perform a \textit{rotation} of the edge between current node and its parent (if not the root).
 \end{itemize}
 Whenever the pointer moves to or it is initialized to a node $v$, we say that node $v$ is \textit{touched}. A BST-model search algorithm  is correct if during each search, the element $x_i$ that is being searched for is touched. The cost of a BST algorithm on the search sequence $X$ equals the total number of unit-cost operations performed to execute the searches in the sequence. This model was formally defined in ~\cite{DBLP:journals/siamcomp/DemaineHIP07} and it is known to be equivalent up to constant factors to several alternative models which have been considered (e.g.~\cite{DBLP:conf/soda/DemaineHIKP09,DBLP:journals/siamcomp/Wilber89}).

 
 A BST data structure can be augmented such that each node stores $O(\log N)$ additional bits of information. The running time of such BST data structures in the RAM model is dominated by the number of unit-cost operations.
%
  A \textit{static} BST  is a restricted version of the BST model where rotations are not allowed and thus the shape of the tree never changes.

 \subsection{The B-tree model} 
 
 We define the B-tree model to be a generalization of the BST model which allows more than one key to be stored in each node. The B-tree model is parameterized by a positive integer $B\geq 2$ which represents the maximum number of children of each node\footnote{Recall that in the external memory model (defined in Section~\ref{s:intro}) B denotes the block size. Each B-tree node has at most $B$ children, contains $O(B)$ words and thus it can be stored in $O(1)$ blocks of size B.}; in the case where $B=2$ the B-tree model will be equivalent to the BST model. We denote by $n(v)$ the number of keys stored in a node $v$. Every node $v$ has $n(v)\leq B-1$ and $n(v)+1$ child pointers (some of which could be null). A node $v$ which stores exactly $n(v) = B-1$ keys is called \textit{full}. 
 
 Suppose $x_1,\dotsc,x_{n(v)}$ are the keys stored at node $v$ and $c_1,\cdots,c_{n(v)+1}$ are the children of $v$. 
 Keys satisfy the in-order condition, i.e. $x_1 < \dotsc < x_{n(v)}$ and for any key $k_i$ stored in the subtree $T_{c_i}$ rooted at $c_i$, we have that $ k_1 < x_1 < \cdots < k_i < x_i < k_{i+1}< \cdots < k_{n(v)} < x_{n(v)} < k_{n(v)+1}   $. 
 
 
 Similar to the BST model, to execute each search there is a single pointer initialized to the root of the tree at unit cost. To execute a search for $x_i$, a B-tree algorithm performs a sequence of the following unit-cost operations which are described formally later:
 
 \begin{itemize}
  \item Move the pointer to a child or to the parent of the current node.
  \item Split a node containing at least three keys.
  \item Join two sibling nodes storing no more than $ B-2 $ keys in total.
  \item Rotate the edge between the current node and its parent.
 \end{itemize}
 
B-tree model algorithms that only use the first type of operations are referred  to as \emph{static} as the shape of the B-tree does not change.
 We now fully describe the unit-cost operations of rotating, splitting and joining:  
 
 \begin{description}
 \item[Rotations:] Consider a (non-root) node $u$ and let $p(u)$ be its parent. Let $P =\lbrace p_1, \dotsc, p_m \rbrace$ be the union of all keys stored in $u$ and $p(u)$. The keys stored at $u$ define an interval $[p_{\ell},p_r]$ in $P$. A rotation of the edge $(p(u),u)$ essentially updates this interval to $[p_{\ell'},p_{r'}]$, moving the keys as needed. Depending on the values of $\ell, \ell'$ and $r, r'$ we characterize a rotation as a promote/demote left --- promote/demote right rotation. For example, a rotation of the type promote left $k$ --- demote right $k'$ sets $\ell' = \ell + k$ (i.e. the $k$ leftmost keys of $u$ are promoted to $p(u)$) and $r' = r + k'$ (i.e. keys $p_{r+1},\dotsc,p_{r+k'}$ are demoted to $u$). Values $k$ and $k'$ should be non-negative and satisfy that after the rotation both $u$ and $p(u)$ have at most $B-1$ keys. Rotations of the type demote left --- promote right, promote left --- promote right and demote left --- demote right can be defined analogously. As an example, figure~\ref{fig:rotation} shows a rotation of type demote left - promote right.

\item[Splitting a node:] Let $u$ be a node (except the root) containing at least three keys and let $p(u)$ be its non-full parent. Splitting node $u$ at key $u_m$ (which is not the smallest or the largest key stored at $u$) consists of promoting $u_m$ to $p(u)$ and replacing $u$ by 2 nodes $u_L, u_R$ such that keys smaller than $u_m$ are contained in $u_L$ and keys larger than $u_m$ are in $u_R$. To split the root (given that it stores at least three keys), we create an empty B-tree node, make it the parent of the root (i.e. the new root) and then perform a split operation as defined above.

\item[Join:] This operation is the inverse of a split. Let $u$ and $v$ be two sibling nodes and let $p$ be their parent, such that there exists a unique key $p_j$ in $p$ such that $p_j$ is larger than all keys stored at $u$ and smaller than all keys stored at $v$. Joining nodes $u$ and $v$ (given that they store no more than $B-2$ keys in total) consists of demoting $p_j$ to $u$ (and deleting it from $p$), adding all elements of $v$ (including the pointers to children) to $u$ and deleting $v$. Note that after a join operation $p$ might become empty (in case $p_j$ was the unique key of $p$). In that case, we set the parent of $u$ to be the parent of $p$ (if it exists) and we delete $p$. If $p$ is empty and it is the root, then we just delete $p$ and $u$ becomes the new root of the tree. 
\end{description}

\begin{figure}[t]
\begin{center}

\includegraphics[scale=.85]{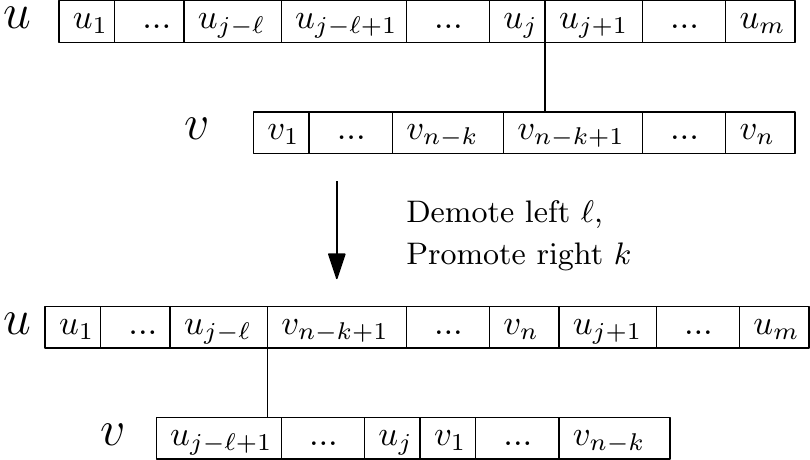} %
\caption{A rotation of a B-tree edge $(u,v)$ of the type demote left $\ell$ --- promote right $k$: From the left of $v$, the $\ell$ neighboring keys of $u$, $ u_{j-\ell+1},\dotsc,u_j $ are getting demoted to $v$. From the right, the $k$ last elements of $v$, $ v_{n-k+1},\dotsc,v_n$ are getting promoted to $u$.}
\label{fig:rotation}
\end{center}
\end{figure}


A  B-tree can be augmented with additional $O(B \log N)$ bits of information for each node. The performance of B-trees in the \textit{external memory model} with blocks of size $B$, is within a  constant factor of the sum of the unit-cost operations as we have defined them.

\paragraph{{\normalfont{\textbf{Relation with other B-tree models.}}}} The classic structure of B-trees first appeared in \cite{DBLP:journals/acta/BayerM72}. In this framework, all leaves have the same depth and no join, split and rotate operations are performed during searches (to be precise, restricted versions of split and join were defined in order to support insertions and deletions and were not allowed for performing search operations, see \cite{CLRS} for an extensive treatment). We call this framework the \textit{classic B-tree} model.
 
 A more flexible model of B-trees was considered in~\cite{DBLP:conf/soda/BoseDL08}: We start with a classic B-tree and an algorithm is allowed to perform joins and splits, but not rotations. Note that by performing join and split operations, the property that all leaves of the tree have the same depth is maintained throughout the whole execution. This model was called  ``self-adjusting B-trees''. To avoid confusion with our dynamic B-tree model, we call this model \textit{classic self-adjusting} B-trees, in order to emphasize that all leaves have the same depth, as in classic B-trees. The self-adjustment relies on the fact that using joins and splits the algorithm might choose to bring an item closer to the root or demote it farther from the root. Also, note that the number of nodes in a B-tree on $N$ keys is not fixed (as opposed to BSTs where we always have exactly $N$ nodes) and the split/join operations might increase/decrease the number of nodes of the tree, changing thus its shape.



 For the rest of this paper, whenever we use the term B-tree we refer to our B-tree model, unless stated otherwise.

\section{Lower bounds: simulating dynamic B-trees using BSTs.}\label{s:lb}

In this section we show how to simulate a dynamic B-tree algorithm using a BST-model algorithm with an $O(\log B)$ overhead in the cost. This will allow us to transform lower bounds from the BST model into lower bounds for the B-tree model.

\paragraph{{\normalfont{\textbf{Notation.}}}} For a search sequence $X$, we denote $\OPTbst(X)$ and $\OPTbtree(X)$ the optimal (offline) cost to serve $X$ using a BST-model and a B-tree-model data structure respectively. 

\begin{theorem}
\label{thm:opt_relation}
For any search sequence $X$,  $\OPTbst(X) = O(\OPTbtree(X) \cdot \log B  )$.
\end{theorem}

\begin{proof}
 We simulate a B-tree execution of $X$ using a BST in the following way: Each node of the B-tree is simulated by a red-black tree of depth $O(\log B)$. Thus our BST is a tree of red-black trees. We also augment the red-black tree data structure such that each node stores a counter on the number of keys in its subtree. Note that in this tree-of-trees, leaves of a red-black tree might have children, which are the roots of other red-black trees. To distinguish the leaves of each tree, we mark the root of each red-black tree. We also use the parent-child terminology for those red-black trees, i.e., if $U$ and $V$ are red-black trees corresponding to B-tree nodes $u$ and $v$ respectively such that $u$ is a child of $v$, we will say that ``tree $U$ is a child of tree $V$''. 
 
 It remains to show that each unit-cost B-tree operation can be simulated in time $O(\log B)$ using our tree-of-trees BST data structure. Moving the pointer from a B-tree node to an adjacent node corresponds to moving the BST pointer from the root of one red-black tree to the root of its child/parent. This can be done in $O(\log B)$ time, since the depth of our red-black trees is $O(\log B)$. For the other unit-cost operations showing this is more complicated. In order to keep the presentation as simple as possible, we proceed as follows: we first describe some basic properties of red-black trees, we then use them to develop operations of merging and separating red-black which will be useful in out tree-of-trees construction and finally we show how to implement the B-tree unit-cost operations using all those tools.

 \paragraph{{\normalfont{\textbf{Background on red-black trees.}}}} We note that red-black trees on $k$ nodes support split and concatenate operations, as well as finding the $\ell$th largest (or smallest key) in time $O(\log k)$~\cite{CLRS}. We now describe those operations. 
\begin{itemize}
 \item The \textit{split} operation of a red-black tree $T$ at a node $x$ re-arranges the tree such that $x$ is the root and the left and right subtrees are red-black trees including keys of smaller and larger values than $x$ respectively. 
 \item \textit{Concatenating} two red-black trees $T_1,T_2$ whose roots are children of a common node $x$, consists of re-arragning the subtree of $x$ to form a red-black tree on all keys of $ T_1 \cup T_2 \cup \lbrace x \rbrace$. This operation is also referred as \textit{concatenating at $x$} and it can be defined even if one of $T_1,T_2$ is empty. Particularly, in our tree of trees construction, if we concatenate at a node $x$ whose left (right) child is marked, then we treat its left (right) subtree as empty. 
 \item \textit{Find the key with a given rank:} Given an augmented red-black tree on $k$ nodes, where each node stores the number of keys in its subtree and a value $\ell < k$, we can find its $\ell$th largest (or smallest) key in $O(\log k)$ time (see e.g.~\cite[Chapter 14]{CLRS}).

 \end{itemize}

 \paragraph{{\normalfont{\textbf{Combining and Separating red-black trees.}}}} We now develop two procedures that will be useful in our implementation of B-tree unit cost operations. In particular we show how to merge and separate red-black trees in $O(\log k)$ time, where $k$ is the total number of nodes in the trees involved. 
 
 \begin{enumerate}[(i)]
  \item \textit{Merge(S,T)}: Given two red-black trees $S$ and $T$ such that $T$ is a child of $S$, merge them into one valid red-black tree. We describe an implementation of this operation in $O(\log k)$ time, where $k$ is the total number of nodes of $S$ and $T$. Let $y_T$ be the root of $T$. We can find the predecessor $\ell$ and the successor $r$ of $y_T$ in $S$ in $O(\log k)$ time, by searching for the key value of $y_T$ in $S$. Note that either $\ell$ or $r$ might not exist. We split $S$ at $\ell$ (if it exists) and then split the right subtree in $r$ (if it exists). Now, $T$ is the left subtree of $r$ (if $r$ does not exist, $T$ is just the right subtree of $\ell$). Unmark the root of $T$. Then, concatenate at $r$ (skip this step if $r$ does not exist) and finally concatenate at $\ell$ (if it exists). The result is a valid red-black tree containing all keys of $S$ and $T$. We used a constant number of $O(\log k)$-time operations. 
  
  \item \textit{Separate(T,$\ell$,$r$):} Given a red-black tree $T$, separate keys with values in the interval $[\ell,r]$, i.e. split $T$ into two trees $T_1, T_2$ where $T_2$ contains keys with values in the interval $[\ell,r]$ and $T_1$ is a parent of $T_2$. In case $\ell$ is not specified ($\ell =null$), we think of $\ell$ as being the minimum key value in $T$ and this operation separates all keys with value at most $r$. Symmetrically, if $r$ is $null$, we think of $r$ as being the maximum key value in $T$ and this operation separates keys with value at least $\ell$. We implement this as follows. Let $\ell'$ be the predecessor of $\ell$ in $T$ (if exists) and $r'$ the successor or $r$ (if exists). Split $T$ in $\ell'$ (skip this step if $\ell'$ does not exist) and then split the subtree with values larger than $\ell'$ at $r'$ (skip this step if $r'$ does not exist). As a result the left subtree of $r'$ (or the right subtree of $\ell'$ if $r'$ does not exists) is the tree $T_2$ containing all keys in $[\ell,r]$. Mark the root of $T_2$. Then concatenate at $r'$ (if exists) and finally concatenate at $\ell'$ (if exists). As a result we get a valid red-black tree $T_1$ which is the parent of red-black tree $T_2$ containing all keys of the interval $[\ell,r]$. 
 \end{enumerate}

 \paragraph{{\normalfont{\textbf{Simulating the unit-cost operations.}}}}
 We now proceed on showing how to simulate B-tree rotations, splits and joins using our tree of red-black trees data structure with cost $O(\log B)$. In all cases, the total number of keys in the trees involved is $O(B)$ and we perform a constant number of operations which take time $O(\log B)$.
 \begin{itemize}
  \item \textbf{Rotations.} We show how to implement a rotation of the form demote left $\ell$ - promote right $k$ (assuming valid values of $\ell$ and $k$). The other operations are defined analogously. Let $(u,v)$ be the B-tree edge which is rotated, where node $u$ is parent of $v$ and let $U$ and $V$ be the augmented red-black trees corresponding to $u$ and $v$. Let $u_1, \dotsc, u_j , u_{j+1}, \dotsc, u_{m}$ and $ v_1, \dotsc, v_{n} $ be the key values stored in $u$ and $v$ respectively such that for all $v_i$ we have that $ u_j < v_i < u_{j+1} $, similar to the example in figure~\ref{fig:rotation}. The rotation corresponds to promoting to $U$ the $k$ largest keys of $V$, i.e. $v_{n-k+1},\dotsc,v_{n}$ and demoting to $V$ the keys $u_{j-\ell+1},\dotsc,u_{j}$.
    We implement such a rotation as follows (see figure~\ref{fig:rotation_b-bst} for an illustration): We start by promoting the $k$ elements to $U$. Find $v_{n-k}$, i.e. the $(k+1)$th largest key stored at $V$. Then, Separate($V,null,v_{n-k}$) to get a tree $V_1$ containing keys $v_{n-k+1},\dotsc,v_{n}$ and a tree $V_2$ with the rest keys of $V$. $V_2$ is a child of $V_1$ and $V_1$ is a child of $U$. Now, we merge $U$ and $V_1$ to get a new tree $U'$, such that $V_2$ is a child of $U'$. It remains to demote $u_{j-\ell+1},\dotsc,u_{j}$ to $V_2$. To do that, we split $U'$ at $v_{n-k+1}$. Let $U_L$ and $U_R$ be the two subtrees of $v_{n-k+1}$ in $U'$. Note that $u_{j-\ell+1},\dotsc,u_{j}$ are the $\ell$ largest keys of $U_L$. Find $u_{j-\ell+1}$, i.e., the $\ell$th largest key of $U_L$ and Separate($U_L,u_{j-\ell+1},null$). We get a separate tree $U_{L_2}$ containing $u_{j-\ell+1},\dotsc,u_{j}$. Mark the root of $U_{L_2}$. Now, $V_2$ is a child of $U_{L_2}$, so we can merge them to form $V''$, the tree corresponding to B-tree node $v$. Finally we concatenate at the root $v_{n-k+1}$, to form the final tree corresponding to $u$, denoted by $U''$, where $V''$ is a child of $U''$.     
  
  \item \textbf{Splitting a node of a B-tree.} Let $u$ be the node which we want to split and $p(v)$ its parent. Let also $U$ and $P$ the corresponding red-black trees, where $U$ is a child of $P$. Let $u_m$ be the median key value of $U$. We split $U$ at $u_m$, so that $u_m$ is the root with subtrees $U_L$ and $U_R$. Mark the roots of $U_L$ and $U_R$ and then merge $u_m$ (which is a single-node red-black tree) with $P$. Clearly all those operations can be performing in $O(\log B)$ time.

  \item \textbf{Joining two sibling nodes.} This is the inverse operation of splitting so the sequence of operations can be seen as the symmetric of the ones performed in splitting. Let $u$ and $v$ be the sibling B-tree nodes that we want to join, and $p$ their parent, with $U$, $V$ and $P$ the corresponding red-black trees in our binary search tree. $U$ and $V$ are children of $P$ and there is a unique key $p_j$ in $P$ such that keys stored at $U$ are smaller than $p_j$ and keys stored at $V$ are larger. Thus, $p_j$ is the successor of the root of $U$ in $P$ and we can find it in $O(\log B)$ time. We then Separate$(P,p_j,p_j)$. Now we get a new tree $P_1$ containing all keys of $P$ except from $p_j$, and $p_j$ is a single-node red-black tree, child of $P_1$. $U$ and $V$ are the left and right children of $p_j$. We unmark the roots of $U$ and $V$ and concatenate at $p_j$, to get a new tree $U'$ and mark its root. Now $U'$ corresponds to the join node of $u$ and $v$, and it is a child of the red-black tree $P'$ which corresponds to the parent node in the B-tree. We performed a constant number of operations each of which takes time $O(\log B)$.
 \end{itemize}
 \qedhere
\end{proof}

\begin{figure}[t]
\begin{center}
\includegraphics[scale= 0.6]{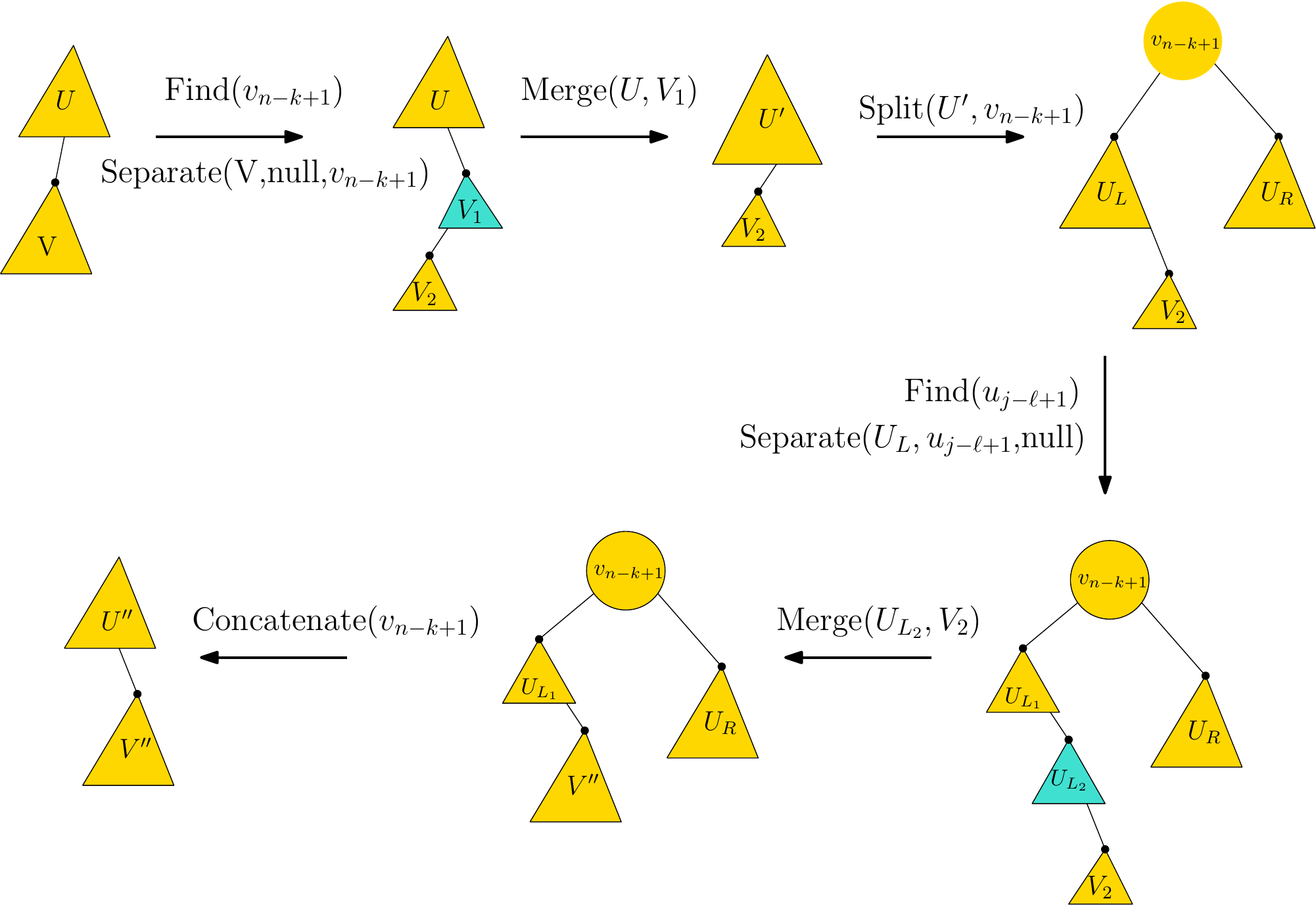}
\caption{Simulating a rotation of a  B-tree edge $(u,v)$ of the type demote left $\ell$ - promote right $k$ in the BST model using red-black tree operations of merge, separate, split, concatenate and find a key with a given rank.}
\label{fig:rotation_b-bst}
\end{center}
\end{figure}

Theorem~\ref{thm:opt_relation} implies that we can transform any lower bound for binary search trees to a lower bound for dynamic B-trees, as shown in the following corollary.

\begin{corol}
\label{cor:lower_bounds}
Let $X$ be a search sequence and let $\LB(X)$ be any lower bound on the cost of executing $X$ in the BST model. 
Then we have that $\OPTbtree(X) = \Omega\lrp{\frac{\LB(X)}{\log B}}$.
\end{corol}

\begin{proof}
Since $\LB(X)$ is a lower bound on $\OPTbst(X)$, we have that
$  \LB(X) \leq \OPTbst(X) = O(\log B) \cdot \OPTbtree(X)$ ,
 which implies $\OPTbtree(X) = \Omega\lrp{\frac{\LB(X)}{\log B}}$. \qedhere
\end{proof}


\begin{onlymain}
\section{Belga B-trees}
\label{s:ods}

In this section, we develop a dynamic B-tree data structure yclept \emph{Belga B-tree} that achieves a competitive ratio of $O(\log \log N)$, for search sequences of length $\Omega(N)$, provided that $1 + \log_B \log N = O(\log_B \log N)$, i.e. $B= (\log N)^{O(1)}$. Our construction is built upon the ideas used in~\cite{DBLP:journals/siamcomp/DemaineHIP07} to get a similar competitive ratio for binary search trees. Particularly, we crucially connect the cost of our algorithm to the \textit{interleave lower bound}. For completeness, we present here the setup and the necessary background regarding this lower bound.

 \paragraph{{\normalfont{\textbf{Interleave Lower Bound and preferred paths (See Figure~\ref{fig:preferred_path}).}}}} Let $\lbrace 1, \ldots ,N \rbrace$ be the keys stored in our B-tree. Let $P$ be a (fixed) complete binary search tree on those keys. For each internal node $v$ in $P$, we define its left region to be $v$ together with the subtree rooted at its left child and its right region to be the subtree rooted at its right child. Node $v$ has a \textit{preferred child}, which is left or right, depending on whether the last search for a node in its subtree was in its left or right region (if no node of the subtree rooted at $v$ has been searched, then $v$ has no preferred child). 
 
 We define a \textit{preferred path} in $P$ as follows: Start from a node that is not the preferred child of its parent (including the root) and perform a walk by following the preferred child of the current node, until reaching a node with no preferred child. Clearly, a preferred path contains $O(\log N)$ keys.

 Note that during a search for a key, the preferred child of some nodes that are ancestors of the node with the key being searched might change. Each change of preferred child, changes also the preferred paths of $P$. For a search sequence $X$, the interleave lower bound $\IB(X)$ equals the total number of changes of preferred child from left to right or from right to left, over all nodes of $P$.
 We use the following lemma of \cite{DBLP:journals/siamcomp/DemaineHIP07}, which is a slight variant of the first lower bound of
 \cite{DBLP:journals/siamcomp/Wilber89}:

 \begin{lemma}[Lemma 3.2 in \cite{DBLP:journals/siamcomp/DemaineHIP07}]
 \label{lem:interleave_preferred_connect}
 The cost to execute $X$ in the BST model is $\Omega(\IB(X))$ if $|X|=\Omega(N)$.
 \end{lemma}

\paragraph{{\normalfont{\textbf{High-level overview of our structure.}}}} We store each preferred path in a balanced classic B-tree. We call such classic B-trees \textit{auxiliary trees}. Our dynamic B-tree will be a tree of classic B-trees. Recall that Lemma~\ref{lem:interleave_preferred_connect} essentially tells us that the number of preferred paths touched during a request sequence is a lower bound on the value of $\OPTbst$. The idea here is to show that for each preferred path touched, and thus unit of lower bound incurred, we can perform search and all update operations (cutting and merging preferred paths) with an overhead factor $O(\log_B \log N ) = O(\frac{\log \log N}{\log B})$. This will imply that we have a dynamic B-tree with cost $O(\frac{\log \log N}{\log B} \cdot \IB(X)) $. This combined with Lemma~\ref{lem:interleave_preferred_connect} and Corollary~\ref{cor:lower_bounds} implies that the cost of our dynamic B-tree data structure is $O(\log \log N) \cdot \OPT_{\Btree}$.

\begin{figure}[t]
\centering
\includegraphics[scale=0.85]{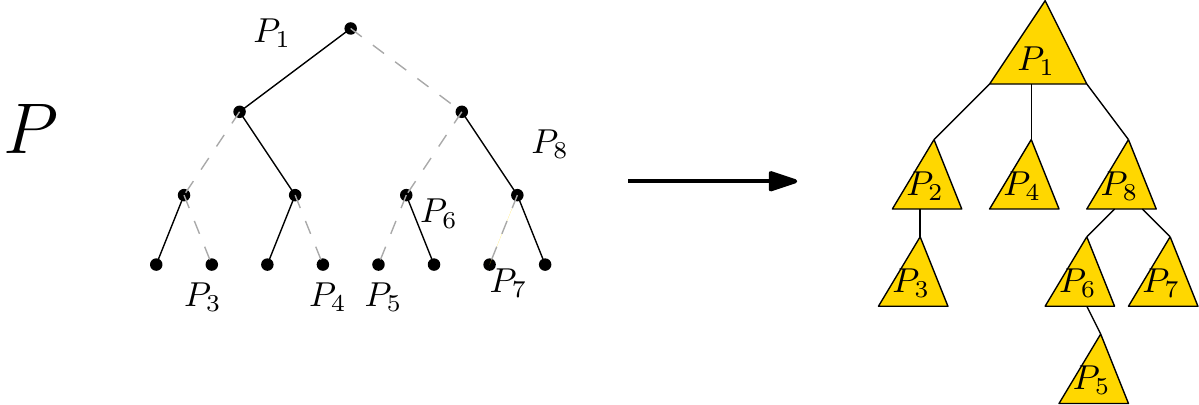}
\caption{Example of a reference tree $P$ and the tree-of-trees representation of its preferred paths $P_1,\dotsc,P_8$. Edges connecting different preferred paths are dashed gray.}
\label{fig:preferred_path}
\end{figure}

\paragraph{{\normalfont{\textbf{Auxiliary trees.}}}} Our auxiliary trees are augmented classic B-trees. Each auxiliary tree stores a preferred path. With each key $x$ we also store its depth in the reference tree $P$. We call this value depth of key $x$. Also, each node stores the minimum and maximum depth of a key in its subtree. Last, a node may be marked or unmarked, depending on whether it is the root of an auxiliary tree or not. Note that $P$ is just a reference tree used for the analysis. We do not need to store $P$ explicitly in order to implement our algorithm. All necessary information about $P$ is stored in our dynamic B-tree data structure. 

During an execution of a search sequence we need to perform the following operations on a preferred path:
\begin{enumerate}[(i)]
 \item Search for a key.
 \item Cut the preferred path into two paths, one consisting of keys of depth at most $d$ and the other of keys of depth greater than $d$.
 \item Merge two preferred paths $P_1$ and $P_2$, where the bottom node of $P_1$ is the parent of the top node of $P_2$.
\end{enumerate}

We will show that we can perform those operations using our auxiliary trees in time $O(1 + \log_B k)$, where $k$ is the number of keys in the involved preferred paths. We defer this proof to the end of this section and we now proceed to the description and analysis of Belga B-trees, assuming that those operations can be done in time $O(1 +\log_B k)$. For the rest of this section, whenever we refer to cutting/merging operations on auxiliary trees, we mean the implementation of cutting/merging the corresponding preferred paths in our B-tree data structure.


\paragraph{{\normalfont{\textbf{Our Algorithm.}}}} A Belga B-tree is a tree of auxiliary classic B-trees, where each auxiliary tree stores a preferred path. Initially we transform the input tree $T_0$ to a valid Belga B-tree. Upon a request for a key $x_i$, we start from the root and search for $x_i$. Whenever we reach a marked node $v$ (i.e. a root of an auxiliary tree), we have to update the preferred paths. Let $Q$ be the preferred path stored in the auxiliary tree of the parent of $v$ and $R$ the preferred path in the auxiliary tree rooted at $v$. We update the preferred paths using the cut and merge operations of auxiliary trees. Particularly, if $d$ is the minimum depth of a key of $R$ (this value is stored at node $v$ of our B-tree), we cut the auxiliary tree storing $Q$ at depth $d-1$. This gives us two preferred paths $Q_{d-}$ and $Q_{d+}$, where the first stores keys of $Q$ of depth at most $d-1$ and the second keys of $Q$ of depth greater than $d$. We mark the roots of the auxiliary trees corresponding $Q_{d-}$ and $Q_{d+}$. We then merge the auxiliary tree storing $Q_{d-}$ with the auxiliary tree rooted at $v$ (which stores $R$). We mark the root of the new tree and continue the search for $x_i$.

Note that the only part where our algorithm needs to perform rotations is the initial step of transforming the input tree into a Belga B-tree.

\paragraph{{\normalfont{\textbf{Bounding the cost.}}}} We now compare the cost of our Belga B-tree data structure to that of the optimal offline B-tree. The following lemma makes the essential connection between the number of preferred paths touched during a search and the cost of our algorithm.

\begin{lemma}
\label{lem:tango_wilber1}
Let $\ell$ be the number of preferred child changes during a search for key $x_i$. Then the cost of Belga B-tree for searching $x_i$ is $O((\ell+1)(1+\log_B \log N))$.  
\end{lemma}

\begin{proof}
To search for $x_i$, we touch exactly $\ell+1$ preferred paths. We account separately for the search cost and the update cost. 

For each preferred path touched, the search cost is $O(\lceil \log_B \log N \rceil)$, since we are searching a balanced B-tree on $O(\log N)$ keys. Thus the total search cost is clearly $O((\ell+1)(1+\log_B \log N))$. 

We now account for the update cost. Recall that we can cut and merge preferred paths on $k$ keys in time $O(1 +\log_B k)$. Since each preferred path has at most $O(\log N)$ keys, we can perform those updates in time $O(1 +\log_B \log N)$. There are $\ell$ preferred path changes, and for each change we perform one cut and and one merge operation, we get that the total time for merging and cutting is $O(\ell \cdot (1 +\log_B \log N))$.  The lemma follows. 
\end{proof}

We now combine this lemma with Corollary~\ref{cor:lower_bounds} to get the competitive ratio of Belga B-trees.

\begin{theorem}
\label{thm:tango_comp}
For any search sequence of length $m=\Omega(N)$, Belga B-trees are $O(\log \log N)$-competitive.
\end{theorem}

\begin{proof}
We account only for the cost occured during searches, since the cost of transforming the input tree into a Belga B-tree is just a fixed additive term which does not depend on the input sequence.

The total number of preferred path changes is at most $\IB(X)+N$. The additive $N$ accounts for the fact that initially each node has no preferred child, so its first change from null to either left or right is not counted in $\IB(X)$. Using Lemma~\ref{lem:tango_wilber1} and summing up over all search requests, we get that the cost of Belga B-trees is $O((\IB(X)+N+m)(1+\log_B \log N))$. By our assumption on the value of $B$, we have that $1 + \log_B \log N = O(\log_B \log N)$, thus the cost is in $O((\IB(X)+N+m) \cdot \frac{\log \log N}{\log B})$. By Lemma~\ref{lem:interleave_preferred_connect} this is bounded by $(\OPTbst(X)+N+m) \cdot \frac{\log \log N}{\log B} $. Using Corollary~\ref{cor:lower_bounds} we get that cost of Belga B-tree is
\[  O\lrp{(\log B \cdot \OPTbtree+N+m) \cdot \frac{\log \log N}{\log B}}.  \]

Note that for any request sequence $\OPTbtree \geq m$. Since $m=\Omega(N)$, we have that $ \log B \cdot \OPTbtree + N +m = O(\log B \cdot \OPTbtree)$. We get that the total cost is upper bounded by \[  O \lrp{ \log B \cdot \OPTbtree \cdot \frac{\log \log N}{\log B}  }  = O(\OPTbtree \cdot \log \log N). \] 
\end{proof}

\paragraph{{\normalfont{\textbf{Operations on auxiliary trees in logarithmic time.}}}}We now show that our auxiliary B-trees support search, cut and merge in time $O(1 + \log_B k)$, where $k$ is the total number of nodes in the trees which are involved.

Before proceeding to this proof we note that classic B-trees on $k$ nodes support search, split and concatenate (similar to the ones we presented in previous section for red-black trees) operations in time $O(1 +\log_B k)$ (see~\cite{CLRS}, Chapter 18). For completeness we describe here the split and concatenate operations:

\begin{itemize}
 \item Splitting a B-tree at a key value $x$ consists of creating a tree where the root contains only $x$, its left subtree is a B-tree on keys with value smaller than $x$ and the right subtree is a B-tree on keys greater than $x$. 
 \item Concatenating two classic B-trees $T_1,T_2$ with a key value $k$ such that all keys in $T_1$ are smaller than $k$ and all keys in $T_2$ are greater, consists of creating a new classic B-tree $T$ which contains all key values contained in $T_1$, $T_2$ and $k$.
\end{itemize}

Search can be clearly performed in time $O(1 + \log_B k)$. We now describe the cut and merge operations on preferred paths.

\textit{Cut a preferred path at depth $d$:} Let $R$ be the tree storing the preferred path. Let $\ell$ and $r$ be the smallest and the largest key value respectively stored at depth greater than $d$ in the path. We wish to find $\ell$ and $r$ in the tree $R$. This can be easily done using the maximum depth value of subtree stored in the nodes. We show how to find $\ell$ and for $r$ is symmetric. Start from the root and move to the leftmost child whose maximum depth is greater than $d$. When we reach a node $v$ such that all its children have maximum depth smaller than $d$, then $\ell$ is the smallest key in $v$ with depth greater than $d$. Let $\ell'$ predecessor of $\ell$ in $R$ (if it has one) and $r'$ the successor of $r$ in $R$ (if it has one). Split $R$ at $\ell'$ (skip this step if $\ell'$ does not exist) and then split the right subtree at $r'$ (skip this step if $r'$ does not exist). Now, the left subtree of $r'$ contains all keys with depth greater than $d$. Let us call this tree $D$. Mark the root of $D$ (and change values of depths,  max depth, min depth in time $O(1+ \log_B k)$) and then use concatenate operations at the tree rooted at $r'$ (if it exists) and then at the tree rooted at $\ell'$ (if it exists) to make the remaining of $R$ a valid classic B-tree.


\textit{Merge two preferred paths:} Let $P_1$ and $P_2$ be the preferred paths that we want to merge, where the bottom node of $P_1$ is the parent of the top node of $P_2$. Merging is the inverse operation of a cut. Let $U$ and $V$ be the auxiliary trees storing $P_1$ and $P_2$ respectively, i.e $U$ is a parent of $V$ in our tree-of-trees construction and the key values stored at $U$ have are of smaller depth in $P$ than the key values stored in $V$. Pick a key from the root of $V$ and find its predecessor $\ell$ and its successor $r$ in $U$. Split $U$ in $\ell$ (skip this step if $\ell$ does not exist) and then split the right subtree at $r$ (skip this step if $r$ does not exist). Now the left subtree of $r$ is $V$. Unmark the root of $V$. Then, concatenate at $r$ to get a resulting tree $R$ which is the right subtree of the root $\ell$ (skip this step if $r$ does not exist). Then, concatenate at $\ell$ (if it exists), to get a valid B-tree which contains all keys of $U$ and $V$. In each of the last two steps (if not skipped), updates of the values of depth, maximum depth, minimum depth take time $O(1 +\log_B k)$. 

\end{onlymain}

\begin{onlymain}
\section{Transforming any static BST into the B-Tree model} \label{s:tranform}


In this section we focus on static trees, with the goal to simulate a static BST using a static B-tree and achieving a speedup by a factor of $\Theta(\log B)$. 
In the static BST and B-Tree models, all that is allowed in each operation is to move a single pointer around the tree, starting at the root, each time moving to a neighboring node, at unit cost per move. We refer to a sequence of moves of a single pointer as a \emph{walk}.
In particular, given a BST we wish to convert it to a B-Tree so that if a walk in the BST costs $k$, a walk in the B-Tree $T_B$ that touches the same keys costs as little as possible in terms of $k$; $k$ is clearly possible since a BST is a B-tree, but when can we achieve $o(k)$?

We note that the results of this section allow the pointer to move arbitrarily in a static BST/B-tree, i.e., it can visit nodes that are outside the path from the root to the searched node. In the case where only a search path of length $D$ is considered, the worst-case cost has been completely characterized in~\cite{DBLP:journals/algorithmica/DemaineIL15}  as
$\Theta \left( {D \over \lg (1{+}B)} \right)$ when $D = O(\lg N)$, $\Theta\left( {\lg N \over \lg \left(1{+}{B \lg N \over D}\right)} \right)$, when 
$D = \Omega(\lg N)$ and $D = O(B \lg N)$, and
$\Theta\left( {D \over B} \right)$ when $D = \Omega(B \lg N)$.


\paragraph{{\normalfont{\textbf{Block-Connected Mappings.}}}} The most natural approach to achieve our goal is to try to map a static BST $T$ into a static B-tree $T_B$ such that each node of $T_B$ corresponds to a connected subtree of $T$. We call such a mapping $f: T \rightarrow T_B$, \textit{block-connected}. Observe that in order to achieve a $\Omega(\log B)$ speedup for the B-tree model $T_B$, it is necessary that a block-connected mapping $f$ should satisfy that every node at depth $d$ in $T$ is at depth $O(\frac{k}{\log B})$ in $T_B$. However, as we will see, this is not sufficient. 

The next theorem shows that, perhaps surprisingly, this approach fails to give any super-constant factor improvement, given that the mapping is deterministic. Afterwards, we show how to achieve an $\Omega(\log B)$ factor speedup using randomization. 

\begin{theorem}
\label{thm:det-static}
 There does not exist a block-connected mapping $f: T \rightarrow T_B$ such that any walk $P$ on $T$ of length $k$ corresponds to a walk of length $o(k)$ in $T_B$.
\end{theorem}




\begin{proof}
We proceed by contraction. Assume an $f$ and $N=2^i-1$ for some integer $i$, and let $T$ be the perfectly balanced tree with $N$ nodes and thus $\ell=\frac{N+1}{2}$ leaves. Consider some BST model sequence of operations $E$ which is an inorder traversal of $T$. Let $b$ be the number of different blocks (i.e. B-tree nodes) that $f(T)$ stores the leaves of $T$ in, which must be at least $\frac{\ell}{B}$. Let $E'$ be the sequence of operations where the inorder traversal does not recurse whenever it encounters a node stored in the same block as a leaf. $E'$ will still visit all $b$ blocks containing leaves, but its length will be exactly $2b-1$. This happens because the block-connected property ensures that $E'$ will never visit two nodes, both of which are in the same block as a leaf of $T$, as that would imply they would have an LCA also in the block, which would mean $E'$ would not visit them. Thus $E'$ has a BST cost of $2b-1$ and a B-tree cost of $\Theta(b)$, where $b=\Omega(\frac{N}{b})$ which proves the theorem. \qedhere
\end{proof}


\paragraph{{\normalfont{\textbf{Randomized Construction.}}}} Theorem~\ref{thm:det-static} above is based on an adversarial argument and relies crucially on the knowledge of the layout of the B-tree. To overcome this issue, we use randomization.

\begin{theorem}
 \label{thm:static_rand} 
 For any BST $T$, there is a randomized block-connected mapping which produces a static B-tree $T_R$ such that for any walk of length $k$ in $T$, there exists a corresponding walk in $T_R$ with expected cost $O\lrp{\frac{k}{\log B}}$.
\end{theorem}




\begin{proof}
We construct the B-tree $T_R$ as follows. We choose uniformly at random an integer $h$ in $[0,\lfloor \log B \rfloor -1 ]$. The root node of $T_R$ contains the key values of the first $h$ levels of $T$. Then, we build the rest of the tree in a deterministic way, by storing $\lfloor \log B \rfloor - 1$ levels of each subtree in a B-tree node, recursively.
Consider any walk $P$ of $k$ operations on $T$ that starts at the root. We assume that the block containing the root and the current location of the walk are stored in memory. Whenever $P$ passes through an edge $e$ of $T$, the probability that this move corresponds to a unit cost operation equals the probability that the endpoints of $e$ belong to different B-tree nodes in $T_R$ and equals $ 1/\lfloor \log B \rfloor $. 

We thus obtain that the expected cost of the corresponding sequence of operations in in $T_R$ is $ k / \lfloor \log B \rfloor$. Since $\lfloor \log B \rfloor = \Omega(\log B)$ for any $B \geq 3$, we get that the expected cost is $O(\frac{k}{\log B})$. \qedhere

\end{proof}


\end{onlymain}

\section{Open Problems} \label{s:open}

We conclude with some open problems. The first is that our Belga B-trees are $O(\log \log N)$-competitive only when $B=\log^{O(1)} N$, and thus the case of large $B$ where $B=\log^{\omega(1)} N$ remains open. The main impediment is to figure out how to fit multiple preferred paths into one block.

A more general open problem is to resolve the following conjecture: Is it possible to convert any BST-model algorithm into a B-Tree model algorithm such that if an algorithm costs $O(k)$ in the BST model, it costs $O(\frac{k}{\log B}+1)$ in the B-Tree model? Special cases of this theorem, when applied to, for example, splay trees and greedy future, would also be interesting should the general conjecture prove too difficult to resolve.

A third open problem is whether, given two B-tree model algorithms, can you achieve the runtime that is the minimum of them; this would be the B-Tree model analogue of the BST result of \cite{DBLP:conf/icalp/DemaineILO13}. It would also allow one to then combine Belga B-trees with other B-tree model algorithms to get stronger results, like, for example \cite{DBLP:conf/soda/BoseDL08} to add the working-set bound; in the BST model \cite{DBLP:conf/soda/WangDS06} gave a $O(\log \log N)$-competitive BST with the working set bound.


\bibliographystyle{alpha}

\bibliography{references_tango,Iacono-John} 

\newpage
\appendix




\end{document}